\documentclass[a4paper,UKenglish,cleveref, autoref, thm-restate]{lipics-v2019}
\pdfoutput=1

\title{SMS in PACE 2020} 

\titlerunning{SMS in PACE 2020} 

\author{Tuukka Korhonen}{Department of Computer Science, University of Helsinki, Finland \and \url{https://tuukkakorhonen.com} }{tuukka.m.korhonen@helsinki.fi}{}{}{}

\authorrunning{Tuukka Korhonen} 

\Copyright{Tuukka Korhonen} 



\category{} 

\relatedversion{} 

\supplement{}

\funding{This work has been financially supported by Academy of Finland (grant 322869).}

\acknowledgements{}

\nolinenumbers 

\EventEditors{John Q. Open and Joan R. Access}
\EventNoEds{2}
\EventLongTitle{42nd Conference on Very Important Topics (CVIT 2016)}
\EventShortTitle{CVIT 2016}
\EventAcronym{CVIT}
\EventYear{2016}
\EventDate{December 24--27, 2016}
\EventLocation{Little Whinging, United Kingdom}
\EventLogo{}
\SeriesVolume{42}
\ArticleNo{23}

\hideLIPIcs

\newcommand{\MS}{\Delta}

\newcommand{\CC}{\mathcal{C}}

\newcommand{\td}{\texttt{\textbf{td}}}

\begin{document}

\maketitle

\begin{abstract}
We describe SMS, our submission to the exact treedepth track of PACE 2020.
SMS computes the treedepth of a graph by branching on the \textbf{S}mall \textbf{M}inimal \textbf{S}eparators of the graph.
\end{abstract}

\section{Overview}
SMS is an exact algorithm implementation for computing treedepth, available in~\cite{tuukka_korhonen_2020_3872898} and in \url{https://github.com/Laakeri/pace2020-treedepth-exact}.
SMS was developed for the fifth PACE challenge (PACE 2020).
The main algorithm implemented in SMS is a recursive procedure that branches on minimal separators~\cite{DBLP:journals/dam/DeogunKKM99}.
Two variants of the branching algorithm are implemented, one with a heuristic algorithm for enumerating minimal separators and one with an exact algorithm~\cite{DBLP:conf/sea2/Tamaki19}.
Several lower bound techniques are implemented within the branching algorithm.
Before applying the branching algorithm, preprocessing techniques are applied and a heuristic upper bound for treedepth is computed.

This arXiv version contains an appendix containing proofs for novel techniques used.

\section{Notation}
Let $G$ be a graph with vertex set $V(G)$ and edge set $E(G)$.
The graph $G[X]$ is the induced subgraph of $G$ with vertex set $X$.
The set $N(v)$ is the neighborhood of a vertex $v$ and $N(X)$ is the neighborhood of a vertex set $X$.
The treedepth of $G$ is denoted by $\td(G)$.
A minimal $a,b$-separator of $G$ is a subset-minimal vertex set $S$ such that the vertices $a$ and $b$ are in different connected components of $G[V(G) \setminus S]$.
The set of minimal separators of $G$ for all pairs $a,b \in V(G)$ is denoted by $\MS(G)$ and the set of minimal separators with size at most $k$ by $\MS_k(G)$.
The set of vertex sets of connected components of $G$ is denoted by $\CC(G)$.

\section{The Algorithm}
\subsection{Branching}
SMS is based on the following characterization of treedepth.

\begin{proposition}[\cite{DBLP:journals/dam/DeogunKKM99}]
\label{pro:td_rec}
Let $G$ be a graph.
If $G$ is a clique then $\td(G) = |V(G)|$.
Otherwise $$\td(G) = \min_{S \in \MS(G)} \left( |S| + \max_{C \in \CC(G[V(G) \setminus S])} \td(G[C]) \right).$$
\end{proposition}

Proposition~\ref{pro:td_rec} is implemented as a recursive algorithm that takes a vertex set $X$ as input and computes $\td(G[X])$ by first enumerating the minimal separators of $G[X]$ and then branching from each minimal separator $S$ to smaller induced subgraphs $G[C]$ for each component $C \in \CC(G[X \setminus S])$.
We make use of upper bounds by implementing Proposition~\ref{pro:td_rec} as a decision procedure which, given a vertex set $X$ and a number $k$, decides if $\td(G[X]) \le k$.
Clearly, in this case we may consider only the minimal separators in $\MS_{k-1}(G[X])$.
Moreover, we handle the minimal separators with sizes $k-1$ and $k-2$ as special cases and thus consider only the minimal separators in $\MS_{k-3}(G[X])$ in the main recursion.
A minimal separator $S$ with $|S| = k-1$ such that $\td(G[X \setminus S]) = 1$ must be a vertex cover of $G[X]$ and therefore is a neighborhood of a vertex.
A minimal separator $S$ with $|S| = k-2$ such that $\td(G[X \setminus S]) \le 2$ has also a somewhat special structure, and we handle them with a modification of Berry's algorithm~\cite{DBLP:journals/ijfcs/BerryBC00} for enumerating minimal separators.

\subsection{Enumerating Small Minimal Separators}
SMS spends most of its runtime in a subroutine which given a number $k$ and a graph $G$ enumerates $\MS_k(G)$.
To make use of the fact that heuristic enumeration of small minimal separators is more efficient than exact enumeration, two variants of the main branching algorithm are ran: first a variant using a heuristic minimal separator enumeration algorithm and then a variant using an exact minimal separator enumeration algorithm.

The heuristic enumeration algorithm is a simple modification of Berry's algorithm~\cite{DBLP:journals/ijfcs/BerryBC00}.
The modification prunes all minimal separators with more than $k$ vertices immediately during the execution, outputting a set $\MS'_k \subseteq \MS_k(G)$ in $O(|\MS'_k| n^3)$ time.
As observed in~\cite{DBLP:conf/sea2/Tamaki19}, there are cases in which $\MS'_k \neq \MS_k(G)$.
However, in practice the algorithm seems to often find all small minimal separators on the values of $k$ that are relevant.

As an exact small minimal separator enumeration algorithm we implement the algorithm of Tamaki~\cite{DBLP:conf/sea2/Tamaki19}, including also the optimizations discussed in the paper.
To the best of our knowledge there are no better bounds than $n^{k+O(1)}$ for the runtime of this algorithm.
In practice it appears to usually have only a factor of 2-10 runtime overhead compared to the heuristic algorithm.

In cases when $G[C]$ is a child of $G$ in the recursion, obtained by branching on a minimal separator $N(C) \in \MS(G)$, and $|C| > |V(G)|/2$ we make use of the small minimal separators of $G$ to enumerate the small minimal separators of $G[C]$.
In particular, for all minimal separators $S \in \MS_k(G[C])$, there exists a minimal separator $S' \in \MS_{k+|N(C)|}(G)$ such that $S = C \cap S'$.
Note that in this case $|N(C)|$ is exactly the difference in the values of $k$ in recursive calls on $G[C]$ and $G$, and therefore $\MS_{k+|N(C)|}(G)$ is already enumerated.

\subsection{Lower Bounds}
To avoid unnecessary re-computation, the known upper and lower bounds for $\td(G[X])$ are stored for each handled induced subgraph $G[X]$.
To this end, an open addressing hashtable with linear probing is implemented.
Also, we implement an ad-hoc data structure so that given a vertex set $X$, a vertex set $X' \subset X$ with the highest known lower bound for $\td(G[X'])$ can be found.
This data structure uses the idea of computing subset-preserving hashes by using the intersection $X \cap V'$, where $V'$ is a subset of vertices with size $O(\log n)$, where $n$ is the number of elements in the data structure.
Other implemented algorithms for computing lower bounds on $\td(G[X])$ are the MMD+ algorithm~\cite{DBLP:journals/iandc/BodlaenderK11} which finds large clique minors, a depth-first search algorithm which finds long paths and cycles, and a graph isomorphism hashtable which finds already processed induced subgraphs $G[X']$ that are isomorphic to $G[X]$ and applies the lower bounds of $G[X']$ to $G[X]$.

\subsection{Preprocessing Techniques}
The preprocessing techniques implemented in SMS are \emph{tree elimination} and the kernelization procedures described in~\cite{DBLP:conf/iwpec/KobayashiT16}.
Tree elimination finds a subgraph $G[T]$ such that $G[T]$ is a tree and $|N(V(G) \setminus T)| = 1$, i.e., the subgraph is attached to the rest of the graph only on a single vertex.
Then it uses an exact algorithm to compute a list of length $\td(G[T])$ that characterizes the behavior of $G[T]$ with respect to treedepth of $G$~\cite{DBLP:journals/ipl/Schaffer89}, and replaces $G[T]$ with a construction of $O(\td(G[T])^2)$ vertices whose behavior is the same.
The simplicial vertex kernelization rule from~\cite{DBLP:conf/iwpec/KobayashiT16} is implemented as it is described there, but the shared neighborhood rule is generalized.
In particular, if there are two non-adjacent vertices $u,v \in V(G)$, and the minimum $u,v$-vertex cut is at least $k$, where $k$ is an upper bound for treedepth, then an edge can be added between $u$ and $v$.

\subsection{Upper Bounds}
To compute upper bounds on treedepth we implement a novel heuristic algorithm.
The algorithm first finds a triangulation (chordal completion) $H$ of $G$ using the LB-Triang algorithm~\cite{DBLP:journals/jal/BerryBHSV06} with a heuristic aiming to minimize the number of fill-edges in each step.
Then it uses the branching algorithm, with some additional heuristics making it non-exact, to compute a treedepth decomposition of $H$.
Any treedepth decomposition of $H$ is also a treedepth decomposition of $G$.
The properties of chordal graphs interplay nicely with the branching algorithm: chordal graphs have a linear number of minimal separators and the treewidth of a chordal graph can be computed in linear time~\cite{DBLP:conf/wg/GalinierHP95}.
Moreover, there exists a triangulation $H$ of $G$ with $\td(H) = \td(G)$, because treedepth can be formulated as a completion problem to a graph class that is a subset of chordal graphs~\cite{DBLP:journals/dam/DeogunKKM99}.



\bibliography{paper}

\appendix

\section{Details and Proofs}
For a vertex set $X \subseteq V(G)$ let $G \setminus X = G[V(G) \setminus X]$.
For a vertex $v \in V(G)$ let $N[v] = N(v) \cup \{v\}$.

\begin{lemma}[\cite{DBLP:journals/ijfcs/BerryBC00}]
A minimal separator $S \in \MS(G)$ has at least two full components, i.e., components $C \in \CC(G \setminus S)$ with $S = N(C)$.
\end{lemma}

\subsection{Minimal Separators with $k-1$ Vertices}
In the branching algorithm, minimal separators with $k-1$ vertices are handled as special cases.

\begin{proposition}
\label{pro:sc1}
Let $S$ be a minimal separator of a graph $G$ such that $\td(G \setminus S) \le 1$. 
It holds that $S = N(v)$ for a vertex $v \in V(G)$.
\end{proposition}
\begin{proof}
The set $V(G) \setminus S$ is not empty, so $\td(G \setminus S) = 1$.
A graph with treedepth $1$ is an independent set, so the full components of $S$ are single vertices.
\end{proof}

By Proposition~\ref{pro:sc1}, it can be decided in polynomial time if there is a minimal separator of size $k-1$ that can be used to obtain a treedepth decomposition of depth $k$.

\subsection{Minimal Separators with $k-2$ Vertices}
Handling minimal separators with size $k-2$ is more complicated.
A minimal separator $S$ with $|S| = k-2$ can be used to obtain a treedepth decomposition of depth $k$ if and only if each connected component of $G \setminus S$ is a star.
To find such minimal separators, we find for each vertex $a \in V(G)$ the minimal separators $S$ such that $a$ is in a full component $C_a$ of $S$ and $C_a$ is a star with $C_a \subseteq N[a]$.
The following proposition can be adapted from~\cite{DBLP:journals/ijfcs/BerryBC00}.

\begin{proposition}[\cite{DBLP:journals/ijfcs/BerryBC00}]
\label{pro:berry_gen}
Let $S$ be a minimal separator of a graph $G$ and $C_a$ a full component of $S$ containing $a$.
The separator $S$ is either a separator close to $a$, obtained as $S = N(C)$ for a component $C \in \CC(G \setminus N[a])$, or can generated from a separator $S'$ with a full component $C'_a \subset C_a$ with $a \in C'_a$ by $S = N(C)$ where $C \in \CC(G \setminus (S' \cup N(v)))$ for a vertex $v \in S'$.
\end{proposition}

By Proposition~\ref{pro:berry_gen}, the minimal separators such that $C_a$ is a star can be generated from other minimal separators such that $C_a$ is a star.
Thus we can modify Berry's algorithm to prune the minimal separators for which $C_a$ is not a star.
We also use an optimization that if $G \setminus (N[a] \cup N(N[a]))$ has a component that is not a star, then there is no minimal separator $S$ such that $a$ is in a full component $C_a \subseteq N[a]$ of $S$ and all components of $G \setminus S$ are stars.

\subsection{Inducing Small Minimal Separators}
We use small minimal separators of $G$ to generate small minimal separators of $G[X]$.

\begin{proposition}
Let $G$ be a graph, $X \subseteq V(G)$ and $S \in \MS_k(G[X])$.
There exists a minimal separator $S' \in \MS_{k+|N(X)|}(G)$ such that $S = S' \cap X$.
\end{proposition}
\begin{proof}
Let $S$ be a minimal $a,b$-separator in $G[X]$.
Let $Y = V(G) \setminus N(X)$.
Now $S$ is a minimal $a,b$-separator in $G[Y]$.
Let us add each vertex $v \in (V(G) \setminus Y)$ to $G[Y]$ one by one.
If $v$ can be reached from both $a$ and $b$ in $G[Y \cup \{v\} \setminus S]$, then $S \cup \{v\}$ is a minimal $a,b$-separator of $G[Y \cup \{v\}]$, having the same full components in $G[Y \cup \{v\}]$ as $S$ has in $G[Y]$.
Otherwise $S$ is a minimal separator of $G[Y \cup \{v\}]$.
With this process we obtain a minimal separator $S' \in \MS(G)$ with $S \subseteq S' \subseteq S \cup N(X)$.
\end{proof}

\subsection{Tree Elimination}
We adapt the algorithm for computing the treedepth of a tree in linear time~\cite{DBLP:journals/ipl/Schaffer89} to locally kernelize tree subgraphs of the input graph.
Treedepth can be formulated as a problem of finding a vertex ranking $c : V(G) \rightarrow \{1, \ldots, \td(G)\}$ such that if $c(v) = c(u)$ for distinct vertices $u,v \in V(G)$, then all paths between $u$ and $v$ contain a vertex $w$ with $c(w) > c(v)$.
Now, for a vertex $v$ we can define the set of ranks $c^+(v)$ that can be ``seen'' from $v$ as the ranks $c(u)$ (including $u = v$) such that there is a path from $v$ to $u$ such that all vertices $w$ of the path have $c(w) \le c(u)$.
Let us say that a set of ranks $R$ is smaller than a set of ranks $R'$ if the maximum element of $R \setminus R'$ is smaller than the maximum element of $R' \setminus R$.

\begin{lemma}[\cite{DBLP:journals/ipl/Schaffer89}]
Let $T$ be a tree and $v \in V(T)$.
A vertex ranking $c$ of $T$ such that $c^+(v)$ is the smallest possible can be computed in linear time.
\end{lemma}

Our local kernelization algorithm finds a subset-maximal subgraph $G[T]$ such that $G[T]$ is a tree and is attached to the rest of the graph only on a single vertex $v \in V(T)$, i.e., $N(V(G) \setminus T) = \{v\}$ and then computes a vertex ranking of $G[T]$ such that $c^+(v)$ is the smallest possible.
Then the subgraph $G[T]$ is replaced with a subgraph determined by $c^+(v)$.
This is done by first replacing $v$ with a path $v = w_1, w_2, \ldots, w_{|c^+(v)|}$ of length $|c^+(v)|$, attached to the rest of the graph only on $v$.
Then, a clique with $c^+(v)(i)-1$ vertices is attached to each vertex $w_i$ of the path, where $c^+(v)(i)$ is the $i$-th smallest rank in $c^+(v)$.
Now, $G[T]$ has been replaced with a construction consisting of $|c^+(v)|$ vertex-disjoint cliques whose sizes are the ranks in $c^+(v)$.
The size of this construction is $\sum_{r \in c^+(v)} r = O(\td(G[T])^2)$ vertices.

\begin{proposition}
Let $G$ be a graph and $G'$ be a graph obtained from $G$ by applying the above described procedure.
It holds that $\td(G) = \td(G')$.
\end{proposition}
\begin{proof}
Assume that the procedure was applied on a single tree $G[T]$ with $N(V(G) \setminus T) = \{v\}$.
Let $c_o^+(v)$ be the smallest possible set of ranks seen from $v$ in $G[T]$.

First, consider a vertex ranking $c$ of $G$ and let us construct a vertex ranking $c'$ of $G'$.
Let $c_t^+(v)$ be the set of ranks on $c$ seen from $v$ in $G[T]$.
If the maximum value of $c_t^+(v)$ is the same as the maximum value of $c_o^+(v)$, then let $c'(w_{|c_o^+(v)|})$ be this value.
The clique attached to $w_{|c_o^+(v)|}$ can now be arbitrarily ranked with smaller values.
Now we can remove the maximum value from both $c_t^+(v)$ and $c_o^+(v)$ and remove the vertex $w_{|c_o^+(v)|}$ and the attached clique and continue by induction.
If the maximum value of $c_t^+(v)$ is different than the maximum value of $c_o^+(v)$, then it must be higher than the maximum value of $c_o^+(v)$ because $c_o^+(v)$ is the smallest.
In that case, we can rank $v$ with the maximum value, and now clearly the rest of the structure of cliques can be ranked with smaller values.

Then, consider a vertex ranking $c'$ of $G'$ and let us construct a vertex ranking $c$ of $G$.
Let $c_t^{'+}(v)$ be the set of ranks on $c'$ seen from $v$ in $G'[T']$, where $T'$ is the set of vertices which replaced $T$.
Our strategy is to show that (1) $c_t^{'+}(v)$ is not smaller than $c_o^+(v)$ and (2) given any set of ranks $R$ not smaller than $c_o^+(v)$, we can construct a ranking $c$ of $G[T]$ such that $c^+(v) \subseteq R$.
In particular, we can construct a ranking for the induced subgraph $G[T]$ which does not add any ranks to the set of seen ranks from $v$, and thus this ranking can be plugged into the rest of the ranking for $G$.
To show (1), if the maximum value of $c_t^{'+}(v)$ is higher the maximum value of $c_o^+(v)$ we are done.
Otherwise, the maximum value must be taken by the clique attached to $w_{|c_o^+(v)|}$.
Let us then remove the maximum value from both $c_t^{'+}(v)$ and $c_o^+(v)$ and remove $w_{|c_o^+(v)|}$ and the attached clique and proceed by induction.
To show (2), let $R$ be a set of ranks not smaller than $c_o^+(v)$ and let us construct a ranking $c$ of $G[T]$ such that $c^+(v) \subseteq R$.
If the maximum value of $R$ is higher than the maximum value of $c_o^+(v)$ then let $c(v)$ be the maximum value, now $c^+(v) = \{c(v)\}$ and the rest of the tree can be ranked with the optimal ranking.
Otherwise find the subtree whose root in the optimal ranking of $G[T]$ has the maximum value of $c_o^+(v)$, rank the subtree with the optimal ranking, remove the subtree, and remove the maximum value from both $R$ and $c_o^+(v)$ and proceed by induction.
\end{proof}

The local kernelization and the reconstruction can be implemented in $O((n + m) \log n)$ time.

\subsection{Generalized Shared Neighborhood Rule}
In~\cite{DBLP:conf/iwpec/KobayashiT16} it was shown that if there are non-adjacent vertices $u,v \in V(G)$ with $|N(v) \cap N(u)| \ge k$, where $k$ is an upper bound for treedepth, then an edge can be added between $u$ and $v$.
We generalize this rule.

\begin{proposition}
Let $u$ and $v$ be non-adjacent vertices in a graph $G$.
If the minimum vertex cut between $u$ and $v$ is at least $k$ and $\td(G) \le k$, then $\td(G) = \td(G')$, where $G'$ is a graph obtained from $G$ by adding an edge between $u$ and $v$.
\end{proposition}
\begin{proof}
Consider a treedepth decomposition of $G$ with depth at most $k$ such that $u$ is not an ancestor of $v$ and $v$ is not an ancestor of $u$.
Now, the ancestors of $v$ form a vertex cut between $u$ and $v$ with size at most $k-1$.
Therefore, in any treedepth decomposition of $G$ with depth at most $k$ either $v$ is an ancestor of $u$ or $u$ is an ancestor of $v$.
\end{proof}

\end{document}